\documentclass[12 pt]{amsart}
\usepackage{amsmath, amsthm, amscd, amsfonts, amssymb, graphicx, color}

\newtheorem{theorem}{Theorem}[section]
\newtheorem{lemma}[theorem]{Lemma}
\newtheorem{proposition}[theorem]{Proposition}
\newtheorem{corollary}[theorem]{Corollary}
\theoremstyle{definition}
\newtheorem{definition}[theorem]{Definition}
\newtheorem{example}[theorem]{Example}

\theoremstyle{remark}

\numberwithin{equation}{section}

\begin{document}
\title[$(1-2u^k)$-constacyclic codes over $\mathcal{R}$]{$(1-2u^k)$-constacyclic codes over $\mathbb{F}_p+u\mathbb{F}_p+u^2\mathbb{F}_p+u^{3}\mathbb{F}_{p}+\dots+u^{k}\mathbb{F}_{p}$}
\author[Z. Raza and A. Rana]{Zahid Raza and Amrina Rana}
 \address{Department of Mathematics,College of Sciences, University of Sharjah, UAE.}
\email { zraza@sharjah.ac.ae}
\address{Department of Mathematics, National University of Computer and Emerging Sciences, Lahore Campus, Pakistan.}
\email {  amrina.rana.1@gmail.com}

\begin{abstract}
Let $\mathbb{F}_p$ be a finite field and $u$ be an indeterminate. This article studies $(1-2u^k)$-constacyclic codes over the ring $\mathcal{R}=\mathbb{F}_p+u\mathbb{F}_p+u^2\mathbb{F}_p+u^{3}\mathbb{F}_{p}+\cdots+u^{k}\mathbb{F}_{p}$ where $u^{k+1}=u$. We illustrate the generator polynomials and investigate the structural properties of these codes via decomposition theorem.
\end{abstract}

\maketitle

\section{Introduction}
The study of coding theory was initiated by Blake in 1970. Many authors worked on different contexts of codes i.e., Linear codes, cylic codes, constacylic codes and etc. The breakthrough work on linear codes was done after remarkable paper of Hammon et.al \cite{HA} which showed non-linear binary codes can be constructed from cyclic codes over $\mathbb{Z}_{4}$ via its Gray images. So, cyclic code over the finite rings is one of  the significant kind of algebraic codes.
Constacyclic codes constitute a remarkable generalization of cyclic codes, hence form an important class of linear codes in coding theory. Constacyclic codes have practical applications in engineering, as they can be efficiently encoded using shift registers. They also have rich algebraic structures for efficient error detection and correction which further applies from data networking to satellite communications. There is a vast literature on constacyclic codes over finite fields and their applications for detailed see, \cite{MA,D3,ZW,KY,ZK}. Zhu and Wang investigated $(1-2u)$ Constacylclic codes over ${\mathbb{F}_{p}+u\mathbb{F}_{p}}$ where $v^{2}=v$ in \cite{ZW}. Moreover, Yildiz and Karadeniz \cite{AKY} studied $(1+v)$ Constacyclic codes of odd length over the ring $\mathbb{F}_{p}+u\mathbb{F}_{p}+v\mathbb{F}_{p}+uv\mathbb{F}_{p}$.

This present paper is focussed on the class of constacyclic codes over the ring $\mathcal{R}=\mathbb{F}_{p}+u\mathbb{F}_{p}+u^{2}\mathbb{F}_{p}+u^{3}\mathbb{F}_{p}+\cdots+u^{k}\mathbb{F}_{p}$ where $u^{k+1}=u$. It is the natural generalization of the results given by Mostafansab and Karimi on the $(1-2u^2)$-constacyclic codes over the ring  $\mathbb{F}_p+u\mathbb{F}_p+u^2\mathbb{F}_p$ in \cite{MK}.
Let $\alpha,\beta$ and $\gamma$ be maps from $\mathcal{R}^m$ to $\mathcal{R}^m$ given by
$$\hspace{-1.1cm}\alpha(s_0,s_1,\dots,s_{m-1})=(s_{n-1},s_0,s_1,\dots,s_{m-2})$$
$$\beta(s_0,s_1,\dots,s_{m-1})=(-s_{m-1},s_0,s_1,\dots,s_{m-2})~~~~~ \mbox{and}$$
$$\hspace{4mm}\gamma(s_0,s_1,\dots,s_{m-1})=((1-2u^k)s_{m-1},s_0,s_1,\dots,s_{m-2})$$
respectively.
Let $\mathcal{L}$ be a linear code of length $m$ over $\mathcal{R}$. Then $\mathcal{L}$ is said to be cyclic if $\alpha(\mathcal{L}) =\mathcal{L}$, negacyclic if $\beta(\mathcal{L})=\mathcal{L}$ and $(1-2u^k)$-constacyclic if $\gamma(\mathcal{L})=\mathcal{L}$.

Let $\mathcal{L}$ be a code of length $m$ over $\mathcal{R}$, and $P(\mathcal{L})$ be its polynomial representation, i.e., $$P(\mathcal{L})=\Big\{\sum\limits_{i=0}^{m-1}s_ia^i|(s_0,\dots,s_{m-1})\in\mathcal{L}\Big\}$$
It is easy to see that:
\begin{theorem}
Let $\mathcal{L}$ be a code of length $m$ over $\mathcal{R}$. Then  $\mathcal{L}$ is $(1-2u^k)$-constacyclic iff $P(\mathcal{L})$ is an ideal of $\mathcal{R}[a]/\langle a^m-(1-2u^k)\rangle$.
\end{theorem}
Let $x=(a_0,a_1,\dots,a_{m-1}),$ $y=(b_0,b_1,\dots,b_{m-1})\in\mathcal{R}^m$. Then $x\cdot y=a_0b_0+a_1b_1 +\dots+a_{m-1}b_{m-1}$ is called Euclidean inner product on $\mathcal{R}^m$. The code $L^\bot=\{x\in\mathcal{R}^m| x\cdot y=0 \mbox{ for every } y\in\mathcal{L}\}$ is called the dual code of $\mathcal{L}$.

The mapping  $\Phi:\mathcal{R}\to\mathbb{F}_p^k$ defined by $$a_{0}+ua_{1}+u^{2}a_{2}+\dots+u^{k}a_{k}\mapsto(-a_{k},a_{1},a_{3},\dots,2a_{0}+a_{k})$$ is called Gray map, which can be extended to $\mathcal{R}^m$ as:

$$\hspace{-11.3cm}\Phi:\mathcal{R}^m\to\mathbb{F}_p^{km}$$
$$(s_0,s_1,\dots,s_{m-1})\mapsto(-a_{k}^{0},-a_{k}^{1}-\dots-a_{k}^{m-1},a_{1}^{0},a_{1}^{1},\dots,a_{1}^{n-1},a_{3}^{0},a_{3}^{1},\dots,$$
$$ \ \ \ \ \ \ \ \ \ \ \ \ \ \ \ \ \ a_{3}^{m-1},\dots,2a_{0}^{0}+a_{k}^{0},2a_{0}^{1}+a_{k}^{1},\dots,2a_{0}^{m-1}+a_{k}^{m-1})$$
where $s_i=a_{o}^{i}+ua_{1}^{i}+u^{2}a_{2}^{i}+\dots+u^{k}a_{k}^{i}$, $0\leq i\leq m-1$.

We denote by $\sigma_1,\sigma_2,\sigma_3,$ respectively the following elements of $\mathcal{R}$:
$$\sigma_1=(1-u^k),~~~\sigma_2=2^{-1}(u^{k-1}+u^k),~~~\sigma_3=2^{-1}(-u^{k-1}+u^k)$$
Note that $\sigma_1,\sigma_2$ and $\sigma_3$ are mutually orthogonal idempotents over $\mathcal{R}$ and $\sigma_1+\sigma_2+\sigma_3=1$.
Let $\mathcal{L}$ be a linear code of length $m$ over $\mathcal{R}$. Define
\begin{eqnarray*}
\mathcal{L}_1&=&\{a\in\mathbb{F}_p^m\mid\exists b,c\in\mathbb{F}_p^m,\sigma_1a+\sigma_2b+\sigma_3c\in\mathcal{L}\}\\
\mathcal{L}_2&=&\{b\in\mathbb{F}_p^m\mid\exists a,c\in\mathbb{F}_p^m,\sigma_1a+\sigma_2b+\sigma_3c\in\mathcal{L}\}\\
\mathcal{L}_3&=&\{c\in\mathbb{F}_p^m\mid\exists a,b\in\mathbb{F}_p^m,\sigma_1a+\sigma_2b+\sigma_3c\in\mathcal{L}\}
\end{eqnarray*}
Then $\mathcal{L}_1,\mathcal{L}_2$ and $\mathcal{L}_3$ are also linear codes of same length. Moreover, we can write uniquely as
$\mathcal{L}=\sigma_1\mathcal{L}_1\oplus\sigma_2\mathcal{L}_2\oplus\sigma_3\mathcal{L}_3.$

\section{Results and Discussions}

\begin{theorem}\label{T2}
Let $\gamma$ denote the $(1-2u^k)$-constacyclic shift of $\mathcal{R}^m$ and $\alpha$ be the cyclic shift of $\mathbb{F}^{km}_p$.
If $\Phi$ is the Gray map of $\mathcal{R}^m$ into $\mathbb{F}^{km}_p$, then $\Phi\gamma=\alpha\Phi$.
\end{theorem}
\begin{proof}
Let $\bar{s}=(s_0,s_1,\dots,s_{m-1})\in\mathcal{R}^m$ where $s_i=a_{0}^{i}+ua_{1}^{i}+u^{2}a_{2}^{i}+\dots+u^{k}a_{k}^{i}$,
 and  $a_0^{i},a_1^{i},\dots,a_k^{i}\in\mathbb{F}_p$, for $0\leq i\leq m-1$. Then, we have

$$\hspace{-4.6cm}\gamma(\bar{s})=\big((1-2u^k)s_{m-1},s_0,s_1,\dots,s_{m-2}\big)$$
$$=\big(a_{0}^{m-1}-ua_{1}^{m-1}-u^{2}a_{2}^{m-1}-,\dots,+(-2a_{0}^{m-1}-a_{k}^{m-1})u^{k},a_{0}^{0}+ua_{1}^{0}+\dots+u^{k}a_{k}^{o},$$
$$a_{0}^{1}+ua_{1}^{1}+\dots+u^{k}a_{k}^{1},\dots,a_{0}^{m-2}+ua_{1}^{m-2}+\dots+u^{k}a_{k}^{m-2}).$$

Now, applying the Gray map $\Phi$, we obtained

$$\hspace{0cm}\Phi(\gamma(\bar{s}))=\big(2a_{0}^{m-1}+a_{k}^{m-1},-a_{k}^{0},-a_{k}^{1}-\dots-a_{k}^{m-2},2a_{0}^{m-1}+(-2a_{0}^{m-1}-a_{k}^{m-1})$$
$$2a_{0}^{0}+a_{k}^{0},2a_{0}^{1}+a_{k}^{1},\dots,2a_{0}^{m-1}+a_{k}^{1},\dots,2a_{0}^{m-1}+a_{k}^{m-1}).$$
On the other hand,
\begin{eqnarray*}
\alpha(\Phi(\bar{s}))&=&\alpha(-a_{k}^{0},-a_{k}^{1},\dots,a_{k}^{m-1},2a_{0}^{0}+a_{k}^{0},2a_{0}^{1}+a_{k}^{1},\dots,2a_{0}^{m-1}+a_{k}^{1},\dots,2a_{0}^{m-1}+a_{k}^{m-1})\\
&=&\big(2a_{0}^{m-1}+a_{k}^{m-1},-a_{k}^{m-1},-a_{k}^{0},-a_{k}^{1},\dots,a_{k}^{m-1},2a_{0}^{0}+a_{k}^{0},2a_{o}^{k}+a_{k}^{1},\dots,2a_{0}^{m-1}+a_{k}^{m-1})
\end{eqnarray*}
Therefore, $\Phi\gamma=\alpha\Phi.$
\end{proof}
\begin{theorem}\label{T3}
$\Phi(\mathcal{L})$ is  a
cyclic code over $\mathbb{F}_p$ of length $3m$.
\end{theorem}
\begin{proof}
Since $\mathcal{L}$ is a $(1-2u^k)$-constacyclic code over $\mathcal{R} \gamma(\mathcal{L})=\mathcal{L}$. Hence, $(\Phi\gamma)(\mathcal{L})=\Phi(\mathcal{L})$. 
By Theorem \ref{T2} we have $\alpha(\Phi(\mathcal{L}))=\Phi(\mathcal{L})\Rightarrow\Phi(\mathcal{L})$ is a cyclic code.
\end{proof}
Notice that $(1-2u^k)^m=1-2u^k$ if $m$ is odd and $(1-2u^k)^m=1$ if $m$ is even.
\begin{proposition}
The code $\mathcal{L}$ is a $(1-2u^k)$-constacyclic code iff  $\mathcal{L}^{\bot}$  is a $(1-2u^k)$-constacyclic.
\end{proposition}
\begin{proof}
The Proposition 2.4 of \cite{D3} grants the ``only if'' part and  $(\mathcal{L}^{\bot})^{\bot}=\mathcal{L}$ gives us the reverse direction.
\end{proof}

A code $\mathcal{L}$ is called self-orthogonal iff $\mathcal{L}\subseteq\mathcal{L}^\bot$.
\begin{proposition}
Let $\mathcal{L}$ be a code of length $m$ over $\mathcal{R}$ s. t. $\mathcal{L}\subset\big(\mathbb{F}_p+u\mathbb{F}_{p}+u^3\mathbb{F}_p+\dots+u^{k}\mathbb{F}_{p}\big)^m$.
If $\mathcal{L}\subseteq\mathcal{L}^\bot$, then  $\Phi(\mathcal{L})\subseteq\Phi(\mathcal{L})^\bot$.
\end{proposition}
\begin{proof}
Assume that $\mathcal{L}$ is self-orthogonal. Let $$s_1=(a_{0}^{1}+ua_{1}^{1}+u^{3}a_{3}^{1}+\dots+u^{k}a_{k}^{1})$$ $$~s_2=(a_{0}^{2}+ua_{1}^{2}+u^{3}a_{2}^{2}+\dots+u^{k}a_{k}^{2})\in\mathcal{L}$$ where $a_{0}^{i},a_{1}^{i},a_{3}^{i},\dots,a_{k}^{i}\in\mathbb{F}_p^m$ for $i=1,2$.
Then
$$s_1\cdot s_2=(a_{0}^{1}+ua_{1}^{1}+u^{3}a_{3}^{1}+\dots+u^{k}a_{k}^{1})\cdot(a_{0}^{2}+ua_{1}^{2}+u^{3}a_{2}^{2}+\dots+u^{k}a_{k}^{2})$$
$$\ \ \ \ \ \ \ =(a_{0}^{1}a_{0}^{2}+(a_{0}^{1}a_{1}^{2}+a_{1}^{1}a_{0}^{2}+\dots+a_{k}^{1}a_{1}^{2})u+(a_{1}^{1}a_{1}^{2}+\dots+a_{k}^{1}a_{k}^{2})u^{2}+$$
$$\ \ \ \ \ \ \ \ \ \ (a_{0}^{1}a_{3}^{2}+a_{1}^{1}a_{3}^{2}+\dots+a_{k}^{1}a_{5}^{2})u^{3}+\dots+(a_{0}^{1}a_{k}^{2}+a_{1}^{1}a_{k}^{2}+\dots+a_{k}^{1}a_{0}^{2})u^k$$
if $s_{1}\cdot s_{2}=0$, then $$a_{0}^{1}a_{0}^{2}=(a_{0}^{1}a_{1}^{2}+a_{1}^{1}a_{0}^{2}+\dots+a_{k}^{1}a_{1}^{2})=(a_{0}^{1}a_{3}^{2}+a_{1}^{1}a_{3}^{2}+\dots+a_{k}^{1}a_{5}^{2})$$
$$=(a_{0}^{1}a_{k}^{2})+a_{1}^{1}a_{k}^{2}+\dots+a_{k}^{1}a_{0}^{2}).$$
Therefore
\begin{eqnarray*}
 \Phi(s_1)\cdot\Phi(s_2)&=&(-a_{k}^{1},a_{1}^{1},a_{3}^{1},\dots,2a_{0}^{1}+a_{k}^{1})\cdot(-a_{k}^{2},a_{1}^{2},\dots,2a_{0}^{2}+a_{k}^{2})\\
 &=&4a_{k}^{1}a_{k}^{2}+a_{1}^{1}a_{1}^{2}+a_{3}^{1}a_{3}^{2}+\dots+2(a_{k}^{1}a_{k}^{2}+a_{0}^{1}a_{k}^{2}+a_{0}^{2}a_{k}^{1}).
\end{eqnarray*}
Hence $\Phi(\mathcal{L}^{\bot})\subseteq\Phi(\mathcal{L})^{\bot}$. Consequently $\Phi(\mathcal{L})\subseteq\Phi(\mathcal{L})^{\bot}.$
\end{proof}

\begin{theorem}\label{T5}
Let $\mathcal{L}=\sigma_1\mathcal{L}_1\oplus\sigma_2\mathcal{L}_2\oplus\sigma_3\mathcal{L}_3$ be a code of length $m$ over $\mathcal{R}$.
Then $\gamma(\mathcal{L})=\mathcal{L}$ iff $\mathcal{L}_1$ is cyclic and
$\mathcal{L}_2$, $\mathcal{L}_3$ are negacyclic codes of length $m$ over $\mathbb{F}_p$.
\end{theorem}

\begin{proof}
First of all, notice that $$(1-2u^k)\sigma_1=\sigma_1,$$ $$(1-2u^k)\sigma_2=-\sigma_2,$$ $$(1-2u^k)\sigma_3=-\sigma_3,$$ Let $\bar{s}=(s_0,s_1,\dots,s_{m-1})\in\mathcal{L}$.
Then $s_i=\sigma_1a_i+\sigma_2b_i+\sigma_3c_i$ where $a_i,b_i,c_i\in\mathbb{F}_p$ $0\leq i\leq m-1$. Let $$a=(a_0,a_1,\dots,a_{m-1}),$$
$$b=(b_0,b_1,\dots,b_{m-1}),$$  $$c=(c_0,c_1,\dots,c_{m-1}).$$ Then $a\in\mathcal{L}_1$, $b\in\mathcal{L}_2$ and $c\in\mathcal{L}_3$.
Assume that $\mathcal{L}_1$ is cyclic and $\mathcal{L}_2$,$\mathcal{L}_3$ are negacyclic codes. Therefore $\alpha(a)\in\mathcal{L}_1$,
$\beta(b)\in\mathcal{L}_2$ and $\gamma(c)\in\mathcal{L}_3$. Thus $\gamma(\bar{s})=\sigma_1\alpha(a)+\sigma_2\beta(b)+\sigma_3\gamma(c)\in\mathcal{L}$.
Consequently $\gamma(\mathcal{L})=\mathcal{L}$.

For the reverse direction, let $a=(a_0,a_1,\dots,a_{m-1})\in\mathcal{L}_1,$ $b=(b_0,b_1,\dots,b_{m-1})\in\mathcal{L}_2,$ $c=(c_0,c_1,\dots,c_{m-1})\in\mathcal{L}_3.$ Set $s_i=\sigma_1a_i+\sigma_2b_i+\sigma_3c_i$ where $0\leq i\leq m-1$.
Hence $\bar{s}=(s_0,s_1,\dots,s_{m-1})\in\mathcal{L}$. Therefore $\gamma(\bar{s})=\sigma_1\alpha(a)+\sigma_2\beta(b)+\sigma_3\gamma(c)\in\mathcal{L}$
which shows that $\alpha(a)\in\mathcal{L}_1$, $\beta(b)\in\mathcal{L}_2$ and $\gamma(c)\in\mathcal{L}_3$. So $\mathcal{L}_1$ is cyclic and
$\mathcal{L}_2$, $\mathcal{L}_3$ are negacyclic codes.
\end{proof}

\begin{theorem}\label{T6}
Let $\mathcal{L}=\sigma_1\mathcal{L}_1\oplus\sigma_2\mathcal{L}_2\oplus\sigma_3\mathcal{L}_3$ be a code of length $m$ over $\mathcal{R}$
s. t. $\gamma(\mathcal{L})=\mathcal{L}$ and  $h_1(a),h_2(a)$ and $h_3(a)$ are the  generating  monic polynomials of $\mathcal{L}_1,\mathcal{L}_2$ and $\mathcal{L}_3$ respectively.
Then $\mathcal{L}=\langle\sigma_1h_1(a),\sigma_2h_2(a),\sigma_3h_3(a)\rangle$ and $|\mathcal{L}|=p^{3m-\sum_{i=1}^3{\rm deg}(h_i)}$
\end{theorem}
\begin{proof}
By Theorem \ref{T5}, $$\mathcal{L}_1=\langle h_1(a)\rangle\subseteq\mathbb{F}_p[a]/\langle a^m-1\rangle,\mathcal{L}_2=\langle h_2(a)\rangle\subseteq\mathbb{F}_p[a]/\langle a^m+1\rangle$$  $$\mathcal{L}_3=\langle h_3(a)\rangle\subseteq\mathbb{F}_p[a]/\langle a^m+1\rangle$$
Since $$\mathcal{L}=\sigma_1\mathcal{L}_1\oplus\sigma_2\mathcal{L}_2\oplus\sigma_3\mathcal{L}_3$$ then
$$\mathcal{L}=\{l(x)|l(x)=\sigma_1g_1(a)+\sigma_2g_2(a)+\sigma_3g_3(a),~ g_1(a)\in\mathcal{L}_1,g_2(a)\in\mathcal{L}_2 \mbox{ and } g_3(a)\in\mathcal{L}_3\}.$$
Hence
$$\mathcal{L}\subseteq\langle\sigma_1h_1(a),\sigma_2h_2(a),\sigma_3h_3(a)\rangle\subseteq\mathcal{R}_m=\mathcal{R}[a]/\langle a^m-(1-2u^k)\rangle.$$
Suppose that $$\sigma_1h_1(a)k_1(a)+\sigma_2h_2(a)k_2(a)+\sigma_3h_3(a)k_3(a)\in\langle\sigma_1h_1(a),\sigma_2h_2(a),\sigma_3h_3(a)\rangle$$ where $k_1(a),k_2(a),k_3(a)\in\mathcal{R}_m$. There exist $$q_1(a)\in\mathbb{F}_p[a]/\langle a^m-1\rangle,q_2(a)\in\mathbb{F}_p[a]/\langle a^m+1\rangle$$ and $q_3(a)\in\mathbb{F}_p[a]/\langle a^m+1\rangle$ such that $$\sigma_1k_1(a)=\sigma_1q_1(a)$$ $$\sigma_2k_2(a)=\sigma_2q_2(a)$$ and $\sigma_3k_3(a)=\sigma_3q_3(a).$
Therefore $$\langle\sigma_1h_1(a),\sigma_2h_2(a),\sigma_3h_3(a)\rangle\subseteq\mathcal{L}.$$ Consequently $\mathcal{L}=\langle\sigma_1h_1(a),\sigma_2h_2(a),\sigma_3h_3(a)\rangle.$ On the other hand $$|\mathcal{L}|=|\mathcal{L}_1|\cdot|\mathcal{L}_2|\cdot|\mathcal{L}_3|=p^{3m-\sum_{i=1}^3{\rm deg}(h_i)}.$$
\end{proof}

\begin{theorem}
Let $\mathcal{L}$ be a code such that $\gamma(\mathcal{L})=\mathcal{L}$.
Then $]!$ polynomial $h(a)$ s. t. $\mathcal{L}=\langle h(a)\rangle$, where $h(a)=\sigma_1h_1(a)+\sigma_2h_2(a)+\sigma_3h_3(a).$
\end{theorem}
\begin{proof}
Suppose that linear codes $\mathcal{L}_1,\mathcal{L}_2$ and $\mathcal{L}_3$ are generated by the monic polynomial $h_1(a),h_2(a)$ and $h_3(a)$ respectively.

By Theorem \ref{T6}, we have $\mathcal{L}=\langle\sigma_1h_1(a),\sigma_2h_2(a),\sigma_3h_3(a)\rangle$.
Let $h(a)=\sigma_1h_1(a)+\sigma_2h_2(a)+\sigma_3h_3(a)$. Then clearly, $\langle h(a)\rangle\subseteq\mathcal{L}$. On the other hand $\sigma_1h_1(a)=\sigma_1h_(a)$ $\sigma_2h_2(a)=\sigma_2h_(a)$ and $\sigma_3h_3(a) =\sigma_3h(a)$, whence $\mathcal{L}\subseteq\langle h(a)\rangle$. Thus $\mathcal{L}=\langle h (a)\rangle.$
The uniqueness of $h(a)$ is followed by that of $h_1(a),h_2(a)$ and $h_3(a)$.
\end{proof}

\begin{lemma}\label{L9}
Let $a^m-(1-2u^k)=h(a)k(a)$ in $\mathcal{R}[a]$ and $\gamma(\mathcal{L})=\mathcal{L}$ be a  code generated
by $h(a)$. If $g(a)$ is relatively prime with $k(a),$ then $\mathcal{L}=\langle h(a)k(a)\rangle$.
\end{lemma}
\begin{proof}
The argument of the proof is same as that of \cite[Lemma 2]{AKY}, so we omit.
\end{proof}

\begin{theorem}\label{T7}
Let $\mathcal{L}=\sigma_1\mathcal{L}_1\oplus\sigma_2\mathcal{L}_2\oplus\sigma_3\mathcal{L}_3$ be a $(1-2u^k)$-constacyclic code of length $m$ over $\mathcal{R}$
such that $h_1(a),h_2(a)$ and $h_3(a)$ are the monic generator polynomials of $\mathcal{L}_1,\mathcal{L}_2$ and $\mathcal{L}_3$ respectively.
Suppose that $$h_1(a)k_1(a)=a^m-1,$$  $$h_2(a)k_2(a)=h_3(a)k_3(a)=a^m+1,$$
 and set $h(a)=\sigma_1h_1(a)+\sigma_2h_2(a)+\sigma_3h_3(a),$  $k(a)=\sigma_1k_1(a)+\sigma_2k_2(a)+\sigma_3k_3(a)$. Then
\begin{enumerate}
\item $h(a)k(a)=a^m-(1-2u^k)$.
\item If $GCD(g_i(a),k_i(a))=1$ for $1\leq i\leq3$, then
$GCD(g(a),k(a))=1$ and $h(a)=h(a)g(a)$ where $g(a)=\sigma_1g_1(a)+\sigma_2g_2(a)+\sigma_3g_3(a)$.
\end{enumerate}
\end{theorem}
\begin{proof}
(1) By assumptions we have
\begin{eqnarray*}
h(a)k(a)&=&h(a)\big(\sigma_1k_1(a)+\sigma_2k_2(a)+\sigma_3k_3(a))\\
&=&\sigma_1h_1(a)k_1(a)+\sigma_2h_2(a)k_2(a)+\sigma_3h_3(a)k_3(a)\\
&=&\sigma_1(a^m-1)+\sigma_2(a^m+1)+\sigma_3(a^m+1)\\
&=&(\sigma_1+\sigma_2+\sigma_3)a^m-(\sigma_1-\sigma_2-\sigma_3)\\
&=&a^m-(1-2u^k).
\end{eqnarray*}
Hence, $h(a)k(a)=a^m-(1-2u^k)$.\\
(2) Suppose that $GCD(g_i(a),k_i(a))=1$ for $1\leq i\leq3$ and let $$g(a)=\sigma_1g_1(a)+\sigma_2g_2(a)+\sigma_3g_3(a)$$ Then for every $1\leq i\leq 3$ there exist $u_i(a),v_i(a)\in\mathcal{R}[a]$ such that $u_i(a)g_i(a)+v_i(a)k_i(a)=1$. Now, set $u(a):=\sigma_1u_1(a)+\sigma_2u_2(a)+\sigma_3u_3(a)$ and $v(a):=\sigma_1v_1(a)+\sigma_2v_2(a)+\sigma_3v_3(a)$. Notice that $\sigma_1+\sigma_2+\sigma_3=1$,\ $\sigma_i^2=1$ and $\sigma_i\sigma_j=0$
for every $1\leq i\neq j\leq3$.
Thus
\begin{eqnarray*}
u(a)g(a)+v(a)k(a)&=&\sigma_1[u_1g_1(a)+v_1(a)k_1(a)]+\sigma_2[u_2(a)g_2(a)+v_2(a)k_2(a)]\\
&+&\sigma_3[u_3(a)g_3(a)+v_3(a)k_3(a)]=\sigma_1+\sigma_2+\sigma_3=1.
\end{eqnarray*}
It follows that $GCD(g(a),k(a))=1$. Now, by part (1) and Lemma \ref{L9}, $$\mathcal{L}=\langle h(a)g(a)\rangle$$ So, the uniqueness of $h(a)$ implies that
$h(a)=h(a)g(a)$.
\end{proof}

Similar to \cite[Theorem 3]{G}, we have the following theorem.
\begin{theorem}
Let $\gamma(\mathcal{L})=\mathcal{L}$  be a code of length $m$ over $\mathcal{R}$. Then
$$\mathcal{L}^\bot=\sigma_1\mathcal{L}_1^\bot\oplus\sigma_2\mathcal{L}_2^\bot\oplus\sigma_3\mathcal{L}_3^\bot.$$
\end{theorem}

\begin{corollary}
Let $\mathcal{L}=\langle\sigma_1h_1(a),\sigma_2h_2(a),\sigma_3h_3(a)\rangle$ be a $(1-2u^k)$-constacyclic code of length $m$ over $\mathcal{R}$ and
$h_1(a),h_2(a)$ and $h_3(a)$ be the monic generator polynomials of $\mathcal{L}_1,\mathcal{L}_2$ and $\mathcal{L}_3$ respectively.
Suppose that $h_1(a)k_1(a)=a^m-1$ and $h_2(a)k_2(a)=h_3(a)k_3(a)=a^m+1$ and
let $k(a)=\sigma_1k_1(a)+\sigma_2k_2(a)+\sigma_3k_3(a)$. then the following conditions hold:
\begin{enumerate}
\item $\mathcal{L}^{\bot}=\langle\sigma_1k_1^{\bot}(a),\sigma_2k_2^{\bot}(a),\sigma_3k^{\bot}_3(a)\rangle$
and $|\mathcal{L}^{\bot}|=p^{\sum_{i=1}^3{\rm deg}(h_i)}$
\item $\mathcal{L}^{\bot}=\langle k^{\bot}(a)\rangle$, $k^{\bot}(a)=\sigma_1k_1^{\bot}(a)+\sigma_2k_2^{\bot}(a)+\sigma_3k^{\bot}_3(a)$
\end{enumerate}
where $k_i^{\bot}(a)$ and  $k^{\bot}(a)$ are the reciprocal polynomial of $k_i(a)$, and $k(a)$ respectively.
\end{corollary}

\begin{theorem}
Let us consider the mapping $\mu:\mathcal{R}[a]/\langle a^m-1\rangle\to\mathcal{R}[a]/\langle a^m-(1-2u^k)\rangle$ defined by
$$\mu\big(l(a)\big)=l\big((1-2u^k)a\big).$$
If $m=$ odd, then $\mu$ is a ring isomorphism.
\end{theorem}
\begin{proof}
Suppose that $u(a)\equiv v(a)$ (mod $a^m-1$). Then there exists $k(a)\in \mathcal{R}[a]$ such that $u(a)-v(a)=(a^m-1)k(a)$.
Therefore
\begin{eqnarray*}
a\big((1-2u^k)a\big)-b\big((1-2u^k)a\big)&=&\big((1-2u^k)^ma^m-1\big)k\big((1-2u^k)a\big)\\
&=&\big((1-2u^k)a^m-(1-2u^k)^2\big)k\big((1-2u^k)a\big)\\
&=&(1-2u^k)\big(a^m-(1-2u^k)\big)k\big((1-2u^k)a\big)
\end{eqnarray*}
which means if $u(a)\equiv v(a)$ (mod $a^m-1)$, then $$u\big((1-2u^k)a\big)\equiv v\big((1-2u^k)a\big)$$ $$\big(mod \ a^m-(1-2u^k)\big)$$
Now, assume that $$u\big((1-2u^k)a\big)\equiv v\big((1-2u^k)a\big)$$ $$\big(mod \ a^m-(1-2u^k)\big).$$ Then there exists $q(a)\in \mathcal{R}[a]$ such that
$$u\big((1-2u^k)a\big)-v\big((1-2u^k)a\big)=\big(a^m-(1-2u^k)\big)q(a).$$
Hence
\begin{eqnarray*}
u(a)-v(a)&=&u\big((1-2u^k)^2a\big)-v\big((1-2u^k)^2a\big)\\
&=&\big((1-2u^k)^ma^m-(1-2u^k)\big)q\big((1-2u^k)a\big)\\
&=&\big((1-2u^k)a^m-(1-2u^k)\big)q\big((1-2u^k)a\big)\\
&=&(1-2u^k)(a^m-1)q\big((1-2u^k)a\big)
\end{eqnarray*}
which means if $u\big((1-2u^k)a\big)\equiv v\big((1-2u^k)a\big)$ $\big($ mod \ $a^m-(1-2u^k)\big)$, then $u(a)\equiv v(a)$ (mod \ $a^m-1)$.
Consequently $u(a)\equiv v(a)$ (mod $a^m-1)\Leftrightarrow u\big((1-2u^k)a\big)\equiv v\big((1-2u^k)a\big)$ $\big($mod
$a^m-(1-2u^k)\big)$ implies hat $\mu$ is well defined and injective, hence $\mu$ is an isomorphism because the rings are finite.
\end{proof}

\begin{corollary}
Let $m$ be an odd natural number. Then the subset $I$ in the ring $\mathcal{R}[x]/\langle a^m-1\rangle$ is an ideal iff
$\mu(I)$ is an ideal of $\mathcal{R}[a]/\langle a^m-(1-2u^k)\rangle.$
\end{corollary}

\begin{corollary}\label{C13}
Let us consider the following permutation of $\mathcal{R}^m$ with $m$ odd s. t.
$$\bar{\mu}(d_0,d_1,\dots,d_{n-1})=(d_0,(1-2u^k)d_1,(1-2u^k)^2d_2,\dots,(1-2u^k)^id_i,\dots,(1-2u^k)^{m-1}d_{m-1})$$
and $\mathcal{D}$ be a subset of $\mathcal{R}^n$. Then $\mathcal{D}$ is a cyclic code iff
$\bar{\mu}(\mathcal{D})$ is a  $(1-2u^k)$-constacyclic code.
\end{corollary}

\begin{definition}
Let us consider the following permutation $\psi$ of the set $\{0,1,\dots,2m-1\}$ with $m$ odd:
$$\psi=(1,m+1)(3,m+3)\cdots(2i+1,m+2i+1)\cdots(m-2,2m-2)$$
Then the  Nechaev permutation is the permutation $\varrho$ of $\mathbb{F}_p^{3m}$ given as 
$$\varrho(d_0,d_1,\dots,d_{2m-1})=(d_{\psi(0)},d_{\psi(1)},\dots,d_{\psi(2m-1)})$$
\end{definition}

\begin{proposition}\label{P15}
Let $\mu$, $\varrho$ is the Nechaev permutation defined as above and $m$ is odd, then $\Phi\bar{\mu}=\varrho\Phi$.
\end{proposition}
\begin{proof}
Let $\bar{s}=(s_0,s_1,\dots,s_i,\dots,s_{m-1})\in\mathcal{R}^m$ where $s_i=a_{o}^{i}+ua_{1}^{i}+u^{2}a_{2}^{i}+\dots+u^{k}a_{k}^{i}$, $0\leq i\leq n-1$. From
$$\bar{\mu}(\bar{s})=(s_0,(1-2u^k)s_1,\dots,(1-2u^k)^is_i,\dots,(1-2u^k)^{m-1}s_{m-1}),$$
it follows that
$$(\varrho\bar{\mu})(\bar{r})=(-a_{k}^{0},2a_{0}^{1}+a_{k}^{1},-a_{k}^{2},2a_{0}^{3}+a_{k}^{3},\dots,2a_{0}^{n-2}+a_{k}^{n-2},-a_{k}^{n-1},a_{0}^{1},a_{1}^{1},$$
$$2a_{0}^{0}+a_{k}^{0},-a_{1}^{1},\dots,a_{1}^{k},a_{3}^{0},a_{3}^{1},\dots,a_{3}^{k},2a_{0}^{2}+a_{k}^{2},\dots,-a_{k}^{n-2},2a_{0}^{n-1}+a_{k}^{n-1})$$
is equal to $(\varrho\Phi)(\bar{r})$.
\end{proof}

\begin{corollary}
Let $m$ be an odd positive integer  and $\varrho$ is the  Nechaev permutation. If $\lambda=\Phi(\mathcal{D})$ where $\mathcal{D}$ is a cyclic code over $\mathcal{R}$, then $\varrho(\lambda)$ is cyclic.
\end{corollary}
\begin{proof}
From Proposition \ref{P15},we have
$(\Phi\bar{\mu})(\mathcal{D})=(\varrho\Phi)(\mathcal{D})=\pi(\lambda)$ and Corollary \ref{C13} gives that $\bar{\mu}(\mathcal{D})$ is a
$(1-2u^k)$-constacyclic code. Thus  by Theorem \ref{T3} $(\Phi\bar{\mu})(\mathcal{D})=\pi(\lambda)$ is a cyclic code.
\end{proof}

Two codes $\mathcal{L}_1$ and $\mathcal{L}_2$ of length $m$ over $\mathcal{R}$ are said to be equivalent if there exists a permutation $w$ of
$\{0,1,\dots,m-1\}$ s. t. $\mathcal{L}_2=\bar{w}(\mathcal{L}_1)$ where $\bar{w}$ is the permutation of the ring $\mathcal{R}^m$ s. t.
$$\bar{w}(d_0,d_1,\dots,d_i,\dots,d_{m-1})=(d_{w(0)},d_{w(1)},\dots,d_{w(i)},\dots,d_{w(m-1)}).$$
\begin{corollary}
The Gray image of a cyclic code over $\mathcal{R}$ of odd length is equivalent to a cyclic code.
\end{corollary}

\begin{example}
Let $m=7$ and $a^7-1=(a-1)(a^3+a+1)(a^3+a^2+1)$ in $\mathcal{R}[a]$.
Applying the ring isomorphism $\mu$, we have $$a^7-(1-2u^k)=(a-(1-2u^k))(a^3+a+(1-2u^k))(a^3+(1-2u^k)a^2+(1-2u^k)).$$
Let $f_1=a-(1-2u^k)$ and $f_2=a^3+a+(1-2u^k)$.
If $\mathcal{L}=(f_1f_2)$, then by Theorem \ref{T3}, the Gray image of the $(1-2u^k)$-constacyclic code $\mathcal{L}$ is a
cyclic code.
\end{example}

\end{document}